\documentclass[conference,10pt]{IEEEtran}

\usepackage{blindtext, graphicx,booktabs}
\usepackage{cite}
\usepackage{amsmath}
\usepackage{amssymb}
\usepackage{subcaption}
\newtheorem{lemma}{Lemma}
\newtheorem{theorem}{Theorem}
\newcommand{\ignore}[1]{}
\begin{document}
\title{A Multi-Class Dispatching and Charging Scheme\\for Autonomous Electric Mobility On-Demand}
\author{\IEEEauthorblockN{Syrine Belakaria$^*$, Mustafa Ammous$^*$, Sameh Sorour$^*$ and Ahmed Abdel-Rahim$^\dag$$^\ddag$}
\IEEEauthorblockA{$^*$Department of Electrical and Computer Engineering,
University of Idaho,
Moscow, ID, USA\\ $^\dag$Department of Civil and Environmental Engineering,
University of Idaho,
Moscow, ID, USA\\$^\ddag$National Institute for Advanced Transportation Technologies, University of Idaho, Moscow, ID, USA\\
Email: \{ammo1375, bela7898\}@vandals.uidaho.edu, \{samehsorour, ahmed\}@uidaho.edu}
}

\maketitle

\begin{abstract}
Despite the significant advances in vehicle automation and electrification, the next-decade aspirations for massive deployments of autonomous electric mobility on demand (AEMoD) services are still threatened by two major bottlenecks, namely the computational and charging delays. This paper proposes a solution for these two challenges by suggesting the use of fog computing for AEMoD systems, and developing an optimized multi-class charging and dispatching scheme for its vehicles. A queuing model representing the proposed multi-class charging and dispatching scheme is first introduced. The stability conditions of this model and the number of classes that fit the charging capabilities of any given city zone are then derived. Decisions on the proportions of each class vehicles to partially/fully charge, or directly serve customers are then optimized using a stochastic linear program that minimizes the maximum response time of the system. Results show the merits of our proposed model and optimized decision scheme compared to both the always-charge and the equal split schemes. 
\end{abstract}

\begin{IEEEkeywords}
Autonomous Mobility On-Demand; Electric Vehicle; Fog-based Architecture; Dispatching; Charging; Queuing Systems.
\end{IEEEkeywords}
\IEEEpeerreviewmaketitle
\section{Introduction}

Urban transportation systems are facing tremendous challenges nowadays due to the dominant dependency and massive increases on private vehicle ownership, which result in dramatic increases in road congestion, parking demand \cite{ref5}, increased travel times \cite{ref6}, and carbon footprint \cite{ref3} \cite{ref4}. These challenges have significantly can be mitigated with the significant advances and gradual maturity of vehicle electrification, automation, and wireless connectivity. With more than 10 million self-driving cars expected to be on the road by 2020 \cite{ref11}, it is strongly forecasted that vehicle ownership will significantly decline by 2025, as it will be replaced by the novel concept of Autonomous Electric Mobility on-Demand (AEMoD) services \cite{ref9,ref10}. In such system, customers will simply need to press some buttons on an app to promptly get an autonomous electric vehicle transporting them door-to-door, with no pick-up/drop-off and driving responsibilities, no dedicated parking needs, no carbon emission, no vehicle insurance and maintenance costs, and extra in-vehicle work/leisure times. With these qualities, AEMoD systems will significantly prevail in attracting millions of subscribers across the world and in providing on-demand and hassle-free private urban mobility.\\ 
\indent Despite the great aspirations for wide AEMoD service deployments by early-to-mid next decade, the timeliness (and thus success) of such service is threatened by two major bottlenecks. First, the expected massive demand of AEMoD services will result in excessive if not prohibitive computational and communication delays if cloud based approaches are employed for the micro-operation of such systems. Moreover, the typical full-battery charging rates of electric vehicles will not be able to cope with the gigantic numbers of vehicles involved in these systems, thus resulting in instabilities and unbounded customer delays. Several recent works \cite{ref1,ref12} have addressed other important problems in autonomous mobility on-demand systems but none of them considered the computational architecture for a massive demand on such services, and the vehicle electrification and charging limitations. \\
\begin{figure}[t]
\centering
      \includegraphics[width=.35\textwidth]{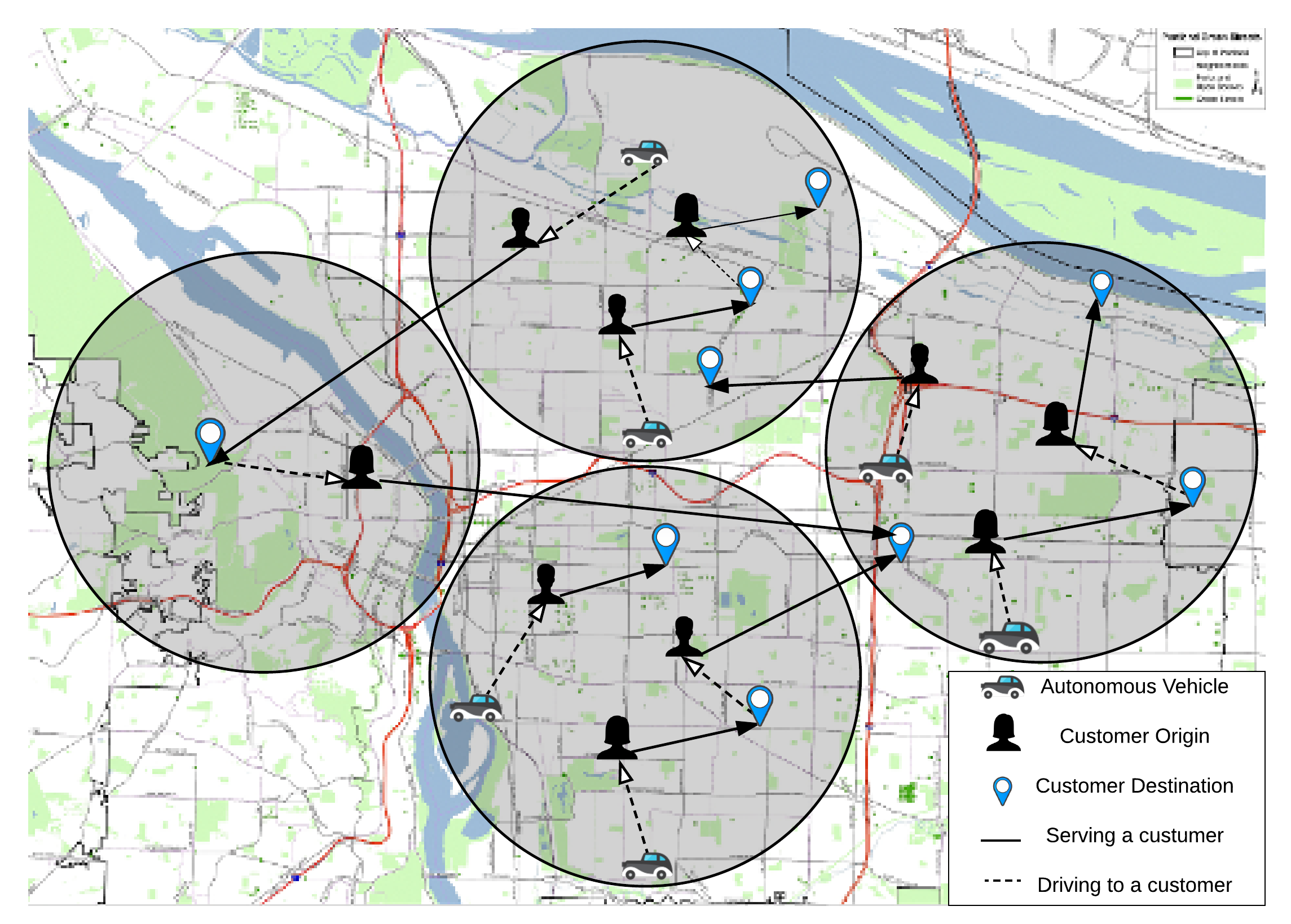}
      \caption{Fog-based architecture for AEMoD system operation}
      \label{fog}
 \end{figure}
\indent In this paper, we suggest to solve the first problem by handling AEMoD system operations in a distributed fashion using fog computing \cite{ref14}. Fog computing is a novel distributed edge computing architecture that pushes computational resources close to the end entities, to provide them with low latency analytics and optimization solutions. The use of fog computing is justified by the fact that many of the AEMoD operations (e.g., dispatching and charging) are localized. Indeed, vehicles located in any city zone are ones that can reach the customers in that zone within a limited time frame. They will also charge in near-by charging points within the zone. Fig.\ref{fog} illustrates a candidate fog-based architecture that can support real-time micro-operational decisions (e.g, dispatching and charging) for AEMoD systems with extremely low computation and communications delays. The fog controller in each service zone is responsible of collecting information about customer requests, vehicle in-flow to the service zone, their state-of-charge (SoC), and the available full-battery charging rates in the service zone. Given the collected information, it can promptly make dispatching, and charging decisions for these vehicles in a timely manner. \\ 
\indent One way to solve the second problem is to smartly cope with the available charging capabilities of each service zone, and efficiently manage the charging options of different SoC vehicles. This former solution can be achieved by introducing the option of partial vehicle charging instead of fully relying on full vehicle charging only when their batteries are depleted. To efficiently manage this partial charging scheme for different SoC vehicles, and motivated by the fact that different customers can be classified in ascending order of their required trip distances (and thus SoC of their allocated vehicles), this paper proposes a multi-class dispatching and charging system of AEMoD vehicles. Arriving vehicles in each service zone are subdivided into different classes in ascending order of their SoC corresponding to the different customer classes. Different proportions of each class vehicles will be then prompted by the fog controller to either wait (without charging) for dispatching to its corresponding customer class or partially charge to serve the subsequent customer class. Even vehicles arriving with depleted batteries will be allowed to either partially or fully charge to serve the first or last class customers, respectively. \\
\indent The question now is: \emph{What is the optimal proportion of vehicles from each class to dispatch or partially/fully charge, to both maintain charging stability and minimize the maximum response time of the system?} To address this question, a queuing model representing the proposed multi-class charging and dispatching scheme is first introduced. The stability conditions of this model and the number of classes that fit the charging capabilities of the service zone are then derived. Decisions on the proportions of each class vehicles to partially/fully charge, or directly serve customers are then optimized using a stochastic linear program that minimizes the maximum response time of the system. Finally, the merits of our proposed optimized decision scheme are tested and compared to both the always-charge and the equal split schemes.

\section{System Model}
We consider one service zone controlled by a fog controller connected to: (1) the service request apps of cutomers in the zone; (2) the AEMoD vehicles; (3) $C$ rapid charging points distributed in the service zone and designed for short-term partial charging; and (4) one spacious rapid charging station designed for long-term full charging. AEMoD vehicles enter the service in this zone after dropping off their latest customers in it. Their detection as free vehicles by the zone's controller can thus be modeled as a Poisson process with rate $\lambda_v $. Customers request service from the system according to a Poisson process. Both customers and vehicles are classified into $n$ classes based on an ascending order of their required trip distance and the corresponding SoC to cover this distance, respectively. From the thinning property of Poisson processes, the arrival process of Class $i$ customers and vehicles, $i \in \{0,\dots, n\}$, are both independent Poisson processes with rates $\lambda_c^{(i)}$ and $\lambda_v p_i$, where $p_i$ is the probability that the SoC of an arriving vehicle to the system belongs to Class $i$. Note that $p_0$ is the probability that a vehicle arrive with a depleted battery, and is thus not able to serve immediately. Consequently, $\lambda^{(0)}_c = 0$ as no customer will request a vehicle that cannot travel any distance. On the other hand, $p_n$ is also equal to 0, because no vehicle can arrive to the system fully charged as it has just finished a prior trip.\\
\ignore{
\begin{equation}\label{eq:1}
\begin{aligned}
\sum ^{n-1}_{i=0}p_i = 1 , \; 0 \leq p_{i} \leq 1  , \; i = 0, \ldots, n-1.
\end{aligned}
\end{equation}
}
\begin{figure}[t]
\centering
 \includegraphics[width=.5\textwidth]{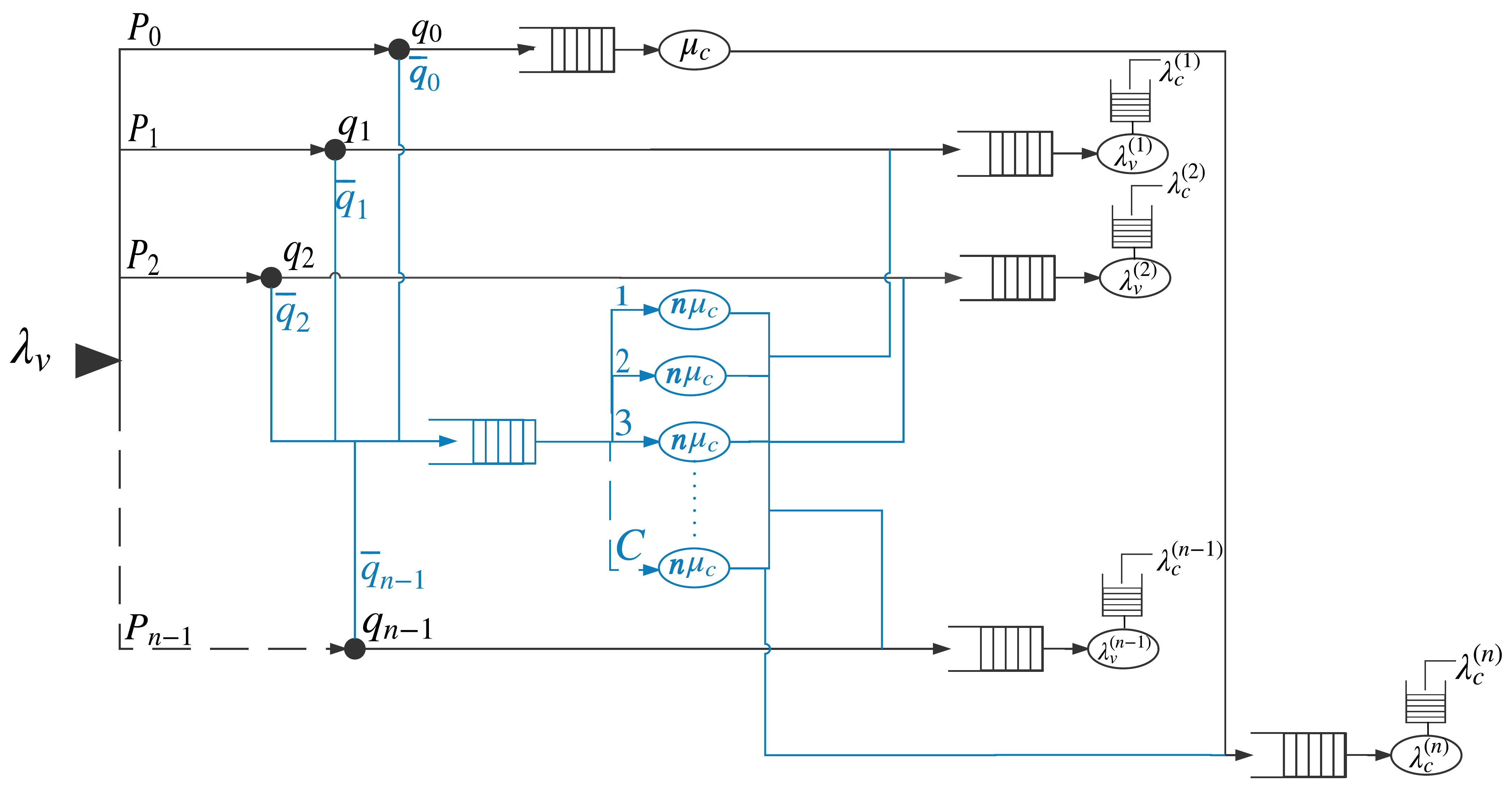}
\caption{Joint dispatching and partially/fully charging model, abstracting an AEMoD system in one service zone.}\label{fig:model}
    \end{figure}
    \ignore{
\begin{table}[t]
\centering
\caption{List of System and Decision Parameters}
\label{vartab}
\begin{tabular}{@{}ll@{}}
\toprule
Variables         & Definition                                                            \\ \midrule
$\lambda_v$       & Total arrival rate of vehicles                                        \\
$p_{i}$           & Probability of arrival of a vehicle from Class $i$      \\
$q_{0}$           & Probability that a battery-depleted vehicle partially charges                          \\
$\overline{q}_0$  & Probability that a battery-depleted vehicle fully charges           \\
$q_{i}, i\neq0$           & Probability that a vehicle in Class $i$ is directly dispatched                          \\
$\overline{q}_i,  i\neq0$  & Probability that a vehicle in Class $i$ partially charges           \\
$\mu_c$           & Service rate of fully charging a battery-depleted vehicle                 \\
$\lambda_v^{(i)}$ & Arrival rate of vehicles of Class $i$            \\
$\lambda_c^{(i)}$ & Arrival rate of customers served by Class $i$'s vehicles \\
$C$ & No. of distributed charging points in the service zone of the fog controller \\ \bottomrule
\end{tabular}
\end{table}
  }
\indent Upon arrival, each vehicle of Class $i$, $i\in\{1,\dots,n-1\}$, will park anywhere in the zone until it is called by the fog controller to either: (1) serve a customer from Class $i$ with probability $q_i$; or (2) partially charge up to the SoC of Class $i+1$ at any of the $C$ charging points (whenever any of them becomes free), with probability $\overline{q}_i=1-q_i$, before parking again in waiting to serve a customer from Class $i+1$. As for Class $0$ vehicles that are incapable of serving before charging,\ignore{ since they cannot serve any customer class without charging,} they will be directed to either fully charge at the central charging station with probability $q_0$, or partially charge at one of $C$ charging points with probability $\overline{q}_0 = 1-q_0$. In the former and latter cases, the vehicle after charging will wait to serve customers of Class $n$ and $1$, respectively.\\
\ignore{
\begin{equation}\label{eq:2}
\begin{aligned}
q_i  + \overline{q_i} =1 , \; 0 \leq q_{i} \leq 1  , \; i = 0, \ldots, n-1.
\end{aligned}
\end{equation}
}
\indent As widely used in the literature (e.g., \cite{ref15,ref16}), the full charging time of a vehicle with a depleted battery is assumed to be exponentially distributed with rate $\mu_c$. Given uniform SoC quantization among the $n$ vehicle classes, the partial charging time can then be modeled as an exponential random variable with rate $n\mu_c$. Note that the larger rate of the partial charging process is not due to a speed-up in the charging process but rather due to the reduced time of partially charging.\ignore{ The use of exponentially distributed charging times for charging electric vehicles has been widely used in the literature \cite{ref15,ref16} to model the randomness in the charging duration of the different battery sizes.} The customers belonging to Class $i$, arriving at rate $\lambda_c^{(i)}$, will be served at a rate of $\lambda_v^{(i)}$, which includes the arrival rate of vehicles that: (1) arrived to the zone with a SoC belonging to Class $i$ and were directed to wait to serve Class $i$ customers; or (2) arrived to the zone with a SoC belonging to Class $i-1$ and were directed to partially charge to be able to serve Class $i$ customers.\\
\indent Given the above description and modeling of variables\ignore{ of the vehicle dispatching and charging variables and options}, the entire zone dynamics can thus be modeled by the queuing system depicted in Fig.\ref{fig:model}. This system includes $n$ M/M/1 queues for the $n$ classes of customer service, one M/M/1 queue for the charging station, and one M/M/C queue representing the partial charging process at the $C$ charging points.\\
\indent Assuming that the service zones will be designed to guarantee a maximum time for a vehicle to reach a customer, our goal in this paper is to minimize the maximum expected response time of the entire system. By response time, we mean the time elapsed between the instant when an arbitrary customer requests a vehicle, and the instant when a vehicle starts moving from its parking or charging spot towards this customer.

\section{System stability conditions}
In this section, we first deduce the stability conditions of our proposed joint dispatching and charging system, using the basic laws of queuing theory. We will also derive an expression for the lower bound on the number $n$ of needed classes that fit the charging capabilities of any arbitrary service zone.
Each of the $n$ classes of customers are served by a separate queue of vehicles, with $ \lambda_v^{(i)} $ being the arrival rate of the vehicles that are available to serve the customers of the $i^{th}$ class. Consequently, it is the service rate of the customers $i^{th}$ arrival queues.
We can thus deduce from Fig. \ref{fig:model} and the system model in the previous section that :
\begin{equation}\label{eq:3}
\begin{aligned}
& & &\lambda_v^{(i)} = \lambda_v(p_{i-1}\overline{q}_{i-1} + p_{i}q_{i}) , \; i = 1, \ldots, n-1.\\
& & &\lambda_v^{(n)} = \lambda_v(p_{n-1}\overline{q}_{n-1} + p_{0}q_{0})\\
\end{aligned}
\end{equation}
Since we know that $\overline{q}_{i} + q_{i} = 1$
Then we substitute $\overline{q}_{i}$ by  $1 - q_{i}$ in order to have a system with $n$ variables
\begin{equation}\label{eq:4}
\begin{aligned}
& & &\lambda_v^{(i)} = \lambda_v(p_{i-1} - p_{i-1}{q_{i-1}} + p_{i}q_{i}) , \; i = 1, \ldots, n-1\\
& & &\lambda_v^{(n)} = \lambda_v(p_{n-1} - p_{n-1}{q_{n-1}} + p_{0}q_{0}) \ignore{, \; i = n}\\
\end{aligned}
\end{equation}
From the well-known stability condition of an M/M/1 queue, we have: 
\begin{equation}\label{eq:5}
\begin{aligned}
& & & \lambda_v^{(i)} > \lambda_c^{(i)} , \; i = 1, \ldots, n 
\end{aligned}
\end{equation}
\indent Before reaching the customer service queues, the vehicles will go through a decision step of either to go to these queues immediately or partially charge. The stability of the charging queues should be guaranteed in order to ensure the global stability of the entire system at the steady state. From the model described in the previous section, and by the well-known stability conditions of M/M/C and M/M/1 queues, we have the following stability constraints on the $C$ charging points and central charging station queues, respectively: 
\begin{equation}\label{eq:6}
\begin{aligned}
& & & \sum ^{n-1}_{i=0}\lambda_v(p_{i} - p_{i}{q_{i}}) < C (n \mu_c) \\
& & &\lambda_v p_{0}q_{0} < \mu_c
\end{aligned}
\end{equation}

The following lemma allows the estimation of the average needed in-flow rate of vehicles for a given service zone given its rate of customer demand on AEMoD services.
\begin{lemma}\label{lem1}
For the entire system stability, the total arrival rate of the customers belonging to all the classes should be strictly less than the total arrival rate of the vehicles 
 \begin{equation}\label{eq:7}
\begin{aligned}
\sum ^{n}_{i=1}{\lambda_c^{(i)}} < \lambda_v
\end{aligned}
\end{equation}
\end{lemma}
\begin{proof}
The proof of Lemma \ref{lem1} is in Appendix A\ignore{ of \cite{ref13}}.
\end{proof}

Furthermore, the following lemma establishes a lower bound on the number of classes $n$, given the arrival rate of the vehicles $\lambda_v$, the full charging rate $\mu_c$, and the number $C$ of partial charging points. 
\begin{lemma}\label{lem2}
For stability of the charging queues, the number of classes $n$ is the system must obey the following inequality:
\begin{equation}\label{eq:8}
\begin{aligned}
n  > \dfrac{\lambda_v}{C \mu_c} - \dfrac{1}{C}
\end{aligned}
\end{equation}
\end{lemma}
\begin{proof}
The proof of Lemma \ref{lem2} is in Appendix B\ignore{ of \cite{ref13}}.
\end{proof}

\section{Joint Charging and Dispatching optimization}
\subsection{Problem Formulation}
The goal of this paper is to minimize the maximum expected response time of the system's classes. The response time of any class is defined as the average of the duration between any customer putting a request until a vehicle is dispatched to serve him/her. The maximum expected response time of the system can be expressed as: 
\begin{equation} \label{eq:9}
\max_{i\in\{1,\dots,n\}} \left\{\dfrac {1}{\lambda_v^{(i)}-\lambda_c^{(i)}}\right\}
\end{equation}
\ignore{
Minimizing $\dfrac {1}{\lambda_v^{(i)}-\lambda_c^{(i)}} $ will lead to the minimization of the average response time of the entire system.The problem here is our objective function will be non-convex which may lead to several solving problems and we will not be able to guarantee an absolute optimal solution.
If we want to have a system in which all the users do not wait more than a certain time T, no matter in which queue they are and no matter how long the trip they are requesting.So we try to find a common upper bound to all the queues delays noted $ T = \dfrac{1}{R} $ so that none of the customers will wait more than it.   
\begin{equation}\label{eq:10}
\begin{aligned}
\dfrac {1}{\lambda_v^{(i)}-\lambda_c^{(i)}} \leq \dfrac{1}{R}, \; R > 0\\
{\lambda_v^{(i)}-\lambda_c^{(i)}} \geq {R} , \; R > 0
\end{aligned}
\end{equation}
This optimization problem can be solved by using a similar methodology to the Piecewise-linear minimization. Given the above constraints, the right decisions to charge the vehicle or serve the customer in order to minimize the average response time for all the customers can be obtained by solving the following optimization problem : 
}

It is obvious that the system's class having the maximum expected response time is the one that have the minimum expected response rate. In other words, we have:
\begin{equation}\label{eq:10}
\arg\max_{i\in\{1,\dots,n\}} \left\{\dfrac {1}{\lambda_v^{(i)}-\lambda_c^{(i)}}\right\} = \arg\min_{i\in\{1,\dots,n\}} \left\{\lambda_v^{(i)}-\lambda_c^{(i)}\right\}
\end{equation}.
Consequently, minimizing the maximum expected response time is equivalent to maximizing the minimum expected response rate. Using the epigraph form \cite{refconvxbook} of the latter problem,\ignore{ and substituting the $\lambda_{v}^{(i)}$ parameters by their expressions in (\ref{eq:4}),} we get the following stochastic optimization problem:
\begin{subequations}\label{eq:11}
\begin{align}
&\qquad\qquad \underset{q_0,q_1,\ldots ,q_{n-1}}{\text{maximize }} R \\
\text{s.t}  & \nonumber\\
&\lambda_v(p_{i-1} - p_{i-1}{q_{i-1}} + p_{i}q_{i}) - \lambda_c^{(i)} \geq R, \; i = 1, \ldots, n-1 \label{eq:11-C1}\\
&\lambda_v(p_{n-1} - p_{n-1}{q_{n-1}} + p_{0}q_{0}) - \lambda^{(n)}_c \geq R \ignore{,\; i = n} \label{eq:11-C2}\\
&\sum ^{n-1}_{i=0}\lambda_v(p_{i} - p_{i}{q_{i}}) < C (n \mu_c) \label{eq:11-C3}\\
&\lambda_v p_{0}q_{0} < \mu_c  \label{eq:11-C4} \\
& \sum ^{n-1}_{i=0}p_i = 1 , \; 0 \leq  p_{i} \leq  1  , \; i = 0, \ldots, n-1  \label{eq:11-C5}\\
&0 \leq q_{i} \leq 1  , \; i = 0, \ldots, n-1  \label{eq:11-C6}\\
&R > 0  \label{eq:11-C7}
\end{align}
\end{subequations}
The $n$ constraints in (\ref{eq:11-C1}) and (\ref{eq:11-C2}) represent the epigraph form's constraints on the original objective function in the right hand side of (\ref{eq:10}), after separation \cite{refconvxbook} and substituting every $\lambda_v^{(i)}$ by its expansion form in (\ref{eq:4}). The constraints in (\ref{eq:11-C3}) and (\ref{eq:11-C4}) represent the stability conditions on charging queues. The constraints in (\ref{eq:11-C5}) and (\ref{eq:11-C6}) are the axiomatic constraints on probabilities (i.e., values being between 0 and 1, and sum equal to 1). Finally, Constraint (\ref{eq:11-C7}) is a positivity constraint on the minimum expected response rate.\ignore{\\
\indent} Clearly, the above equation is a linear program with linear constraints, which can be solved analytically using convex optimization analysis. This will be the target of the next subsection.

\subsection{Optimal Dispatching and Charging Decisions}
The problem in (\ref{eq:11}) is a convex optimization problem with differentiable objective and constraint functions that satisfies Slater's condition. Consequently, the KKT conditions provide necessary and sufficient conditions for optimality. Therefore, applying the KKT conditions to the constraints of the problem and the gradient of the Lagrangian function allows us to find the analytical solution of the decisions $q_i$.
The Lagrangian function associated with the optimization problem in (\ref{eq:11}) is given by the following expression:
\begin{multline}\label{eq:12}
\begin{aligned}
& L(\mathbf{q},R,\boldsymbol{\alpha},\boldsymbol{\beta},\boldsymbol{\gamma},\boldsymbol{\omega}) =
 - R +\sum_{i=1}^{n-1}\alpha_{i}( \lambda_v(p_{i-1}{q_{i-1}} -  p_{i}q_{i} ) \\ 
& + R - \lambda_v p_{i-1} + \lambda_c^{(i)} ) + \alpha_{n} ( \lambda_v( p_{n-1}q_{n-1} - p_{0}{q_{0}} )\\
& + R - \lambda_v p_{n-1} + \lambda_c^{(n)} ) + \beta_{0} (\sum ^{n-1}_{i=0}\lambda_v(p_{i} - p_{i}{q_{i}}) - C (n \mu_c) ) \\ 
& + \beta_{1}(\lambda_v p_{0}q_{0} - \mu_c)+ \sum_{i=0}^{n-1} \gamma_{i}(q_{i} - 1) - \sum_{i=0}^{n-1}\omega_{i}q_{i} + \omega_{n}R\;,
\end{aligned}
\end{multline}
where $\mathbf{q}$ is the vector of dispatching decisions (i.e. $\mathbf{q} = [q_0,\dots,q_{n-1}]$), and where:
\begin{itemize}
\item $\boldsymbol{\alpha} = [\alpha_{i}]$, such that $\alpha_{i}$ is the associated Lagrange multiplier to the $i$-th customer queues inequality.
\item $\boldsymbol{\beta} = [\beta_{i}]$, such that $\beta{i}$ is the associated Lagrange multiplier to the $i$-th  charging queues inequality.
\item $\boldsymbol{\gamma} = [\gamma_{i}]$, such that $\gamma_{i}$ is the associated Lagrange multiplier to the $i$-th  upper bound inequality.
\item $\boldsymbol{\omega} = [\omega_{i}]$, such that $\omega_{i}$ is the  associated Lagrange multiplier to the $i$-th lower bound inequality.
\end{itemize}
\indent By applying the KKT conditions on the equality and inequality constraints, the following theorem illustrates the optimal solution of the problem in (\ref{eq:11}) 
\begin{theorem}\label{thm1}
The optimal charging/dispatching decisions of the optimization problem in (\ref{eq:11}) can be expressed as follows:
\begin{equation}\label{eq:15}
\begin{aligned}
& q_{0}^* = \begin{cases}
            0 &  \alpha_{1}^* > \alpha_{n}^* \\
            1 &  \alpha_{1}^* < \alpha_{n}^* \\
            \end{cases} \\ 
& q_{i}^* =   \begin{cases}
            0 &  \alpha_{i+1}^* > \alpha_{i}^* \\
            1 &  \alpha_{i+1}^* < \alpha_{i}^*  \\
            \end{cases}  , \;  i= 1, \ldots, n-1. \\      
& \text{if } \alpha_{1}^* = \alpha_{n}^* \ne 0 \begin{cases}
             q_{1}^* =\frac{p_{0}{q_{0}^*}}{p_{1}} - \frac{\lambda_v p_{0} - \lambda^{(1)}_c - R^*}{\lambda_v p_{1} }\\
           q_{n-1}^* =\frac{p_{0}{q_{0}^*}}{p_{n-1}} - \frac{\lambda_v p_{0} - \lambda^{(n)}_c - R^*}{\lambda_v p_{n-1} }\\
            \end{cases}  \\       
& \text{if } \alpha_{i+1}^* = \alpha_{i}^* \ne 0
 \begin{cases}
               q_{i}^* =\frac{p_{i-1}{q_{i-1}^*}}{p_{i}} - \frac{\lambda_v p_{i-1} - \lambda_c^{(i)} - R^*}{\lambda_v p_{i} } \\
          q_{i+1}^* =\frac{p_{i}{q_{i}^*}}{p_{i+1}} - \frac{\lambda_v p_{i} - \lambda^{(i+1)}_c - R^*}{\lambda_v p_{i+1} }
            \end{cases}\\
            & i= 1, \ldots, n-1.
\end{aligned}
\end{equation}
\end{theorem}
\begin{proof}
The proof of Theorem \ref{thm1} is in Appendix C\ignore{ in \cite{ref13}}.
\end{proof}

\subsection{Maximum Expected Response Time}
Again, since the problem in (\ref{eq:11}) is convex with convex with differentiable objective and constraint functions, then strong duality holds, which implies that the solution to the primal and dual problems are identical. By solving the dual problem, we can express the optimal value of the maximum expected response time as the reciprocal of the  minimum expected response rate of the system, which is introduced in the following theorem.

\begin{theorem}\label{thm2}
The minimum expected response rate $R^*$ of the entire system can be expressed as: 
 \begin{equation}
 R^* = \sum_{i=1}^{n} \left(\lambda_v p_{i-1} - \lambda_c^{(i)}\right) \alpha_{i}^* + \sum_{i=0}^{n-1} \gamma_i^*
\ignore{
\begin{aligned}
&g(\alpha^*,\beta^*,\gamma^*,\omega^*) = -R^*\\
&R^* = \sum_{i=1}^{n} (\lambda_v p_{i-1} - \lambda_c^{(i)}) \alpha_{i}^* + \sum_{i=0}^{n-1} \gamma_i^*
\end{aligned}
}
\end{equation}
\end{theorem}
\begin{proof}
The proof of Theorem \ref{thm2} is in Appendix D\ignore{ in \cite{ref13}}.
\end{proof}

\section{Simulation Results}
In this section, we test the merits of our proposed scheme using extensive simulations. The simulation metrics used to evaluate these merits are the maximum and average expected response times of the different classes. For all the performed simulation figures, the full-charging rate of a vehicle is set to $\mu_c = 0.033$ mins$^{-1}$, and the number of charging points $C=40$.\\
\indent Fig. \ref{fig:lemmas} illustrates both the interplay of $\lambda_v$ and $\sum ^{n}_{i=1}{\lambda_c^{(i)}}$, established in Lemma \ref{lem1}, and effect of increasing the number of classes $n$ beyond its strict lower bound introduced in Lemma \ref{lem2}.Fig. \ref{fig:lemmas} depict the maximum and average expected response time for different values of $\sum ^{n}_{i=1}{\lambda_c^{(i)}}$, while fixing $\lambda_v$ to 15 min$^{-1}$. For this setting, $n=12$ is the smallest number of classes that satisfy the stability condition in Lemma \ref{lem2}. \\
\indent It is easy to notice that the response times for all values of $n$ increase dramatically when the $\sum ^{n}_{i=1}{\lambda_c^{(i)}}$ approaches $\lambda_v$. Moreover, it shows clearly that increasing $n$ beyond its stability lower bound increases both the maximum and average response times. We thus conclude that the optimal number of classes is the smallest value satisfying Lemma \ref{lem2}:
\begin {small}
\begin{equation}
n^* = \begin{cases}  \dfrac{\lambda_v}{C \mu_c} - \dfrac{1}{C} + 1 &  \dfrac{\lambda_v}{C \mu_c} - \dfrac{1}{C} ~\text{integer}\\
\left\lceil \dfrac{\lambda_v}{C \mu_c} - \dfrac{1}{C}\right\rceil & \text{Otherwise}
\end{cases}
\end{equation}
\end{small}
\indent Fig. \ref{fig:comp-ditributions} depicts the expected response time performances for different distributions of the vehicle SoC and customer trip distances. In this study, we fix $\lambda_v = 8$ and thus $n=7$. We can infer from the figure that the response times for Gaussian distributions of trip distances and both Gaussian or decreasing ones for SoCs are the lowest and exhibit the least response time variance. Luckily, these are the most realistic distributions for both variables. This is justified by the fact that vehicles arrive to the system after trips of different distances, which makes their SoC either Gaussian or slightly decreasing. Likewise, customers requiring mid-size distances are usually more than those requiring very small and very long distances.\\
\indent Fig. \ref{fig:comp-decision} compares the expected response time performances against $\sum ^{n}_{i=1}{\lambda_c^{(i)}}$, for different decision approaches, namely our derived optimal decisions, always partially charge decisions (i.e. $q_i=1~\forall~i$) and equal split decisions (i.e. $q_i=0.5~\forall~i$), for $\lambda_v = 8$ and thus $n=7$. These two schemes represent non-optimized policies, in which each vehicle takes its own fixed decision irrespective of the system parameters. The figure clearly shows superior maximum and average performances for our derived optimal policy compared to the other two policies, especially as $\sum ^{n}_{i=1}{\lambda_c^{(i)}}$ gets closer to $\lambda_v$, which are the most properly engineerd sceanrios (as large differences between these two quantities results in very low utilization). Gains of 13.3\% and 21.3\% on the average and maximum performances can be noticed compared to the always charge policy. This shows the importance of our proposed scheme in achieving better customer satisfaction.
\begin{figure}[t]
\centering
  \includegraphics[width=0.77\linewidth]{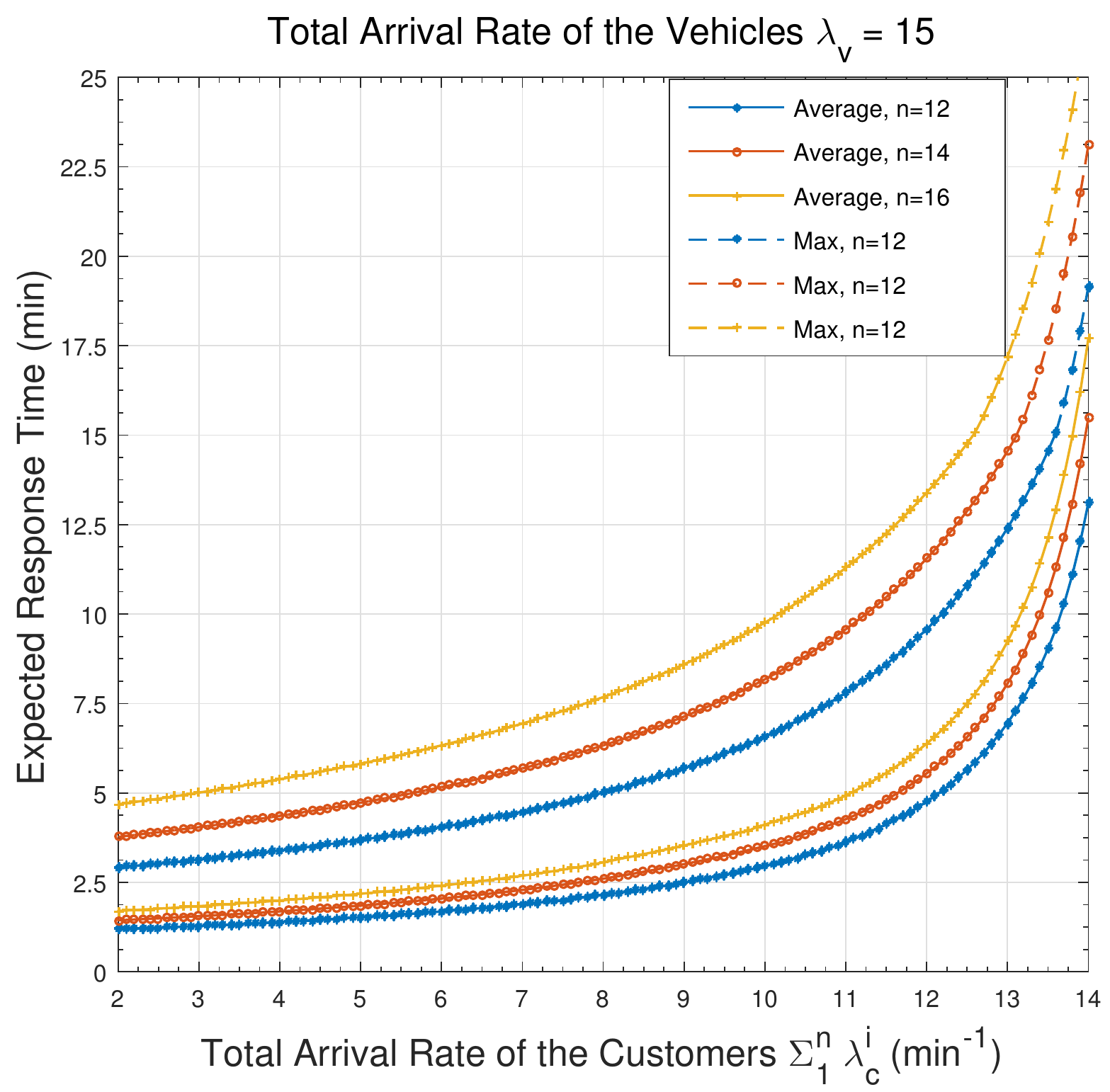}
    \caption{Expected response times for different $\sum ^{n}_{i=1}{\lambda_c^{(i)}}$}\label{fig:lemmas}
\end{figure}
\begin{figure}[t]
\centering
  \includegraphics[width=0.77\linewidth]{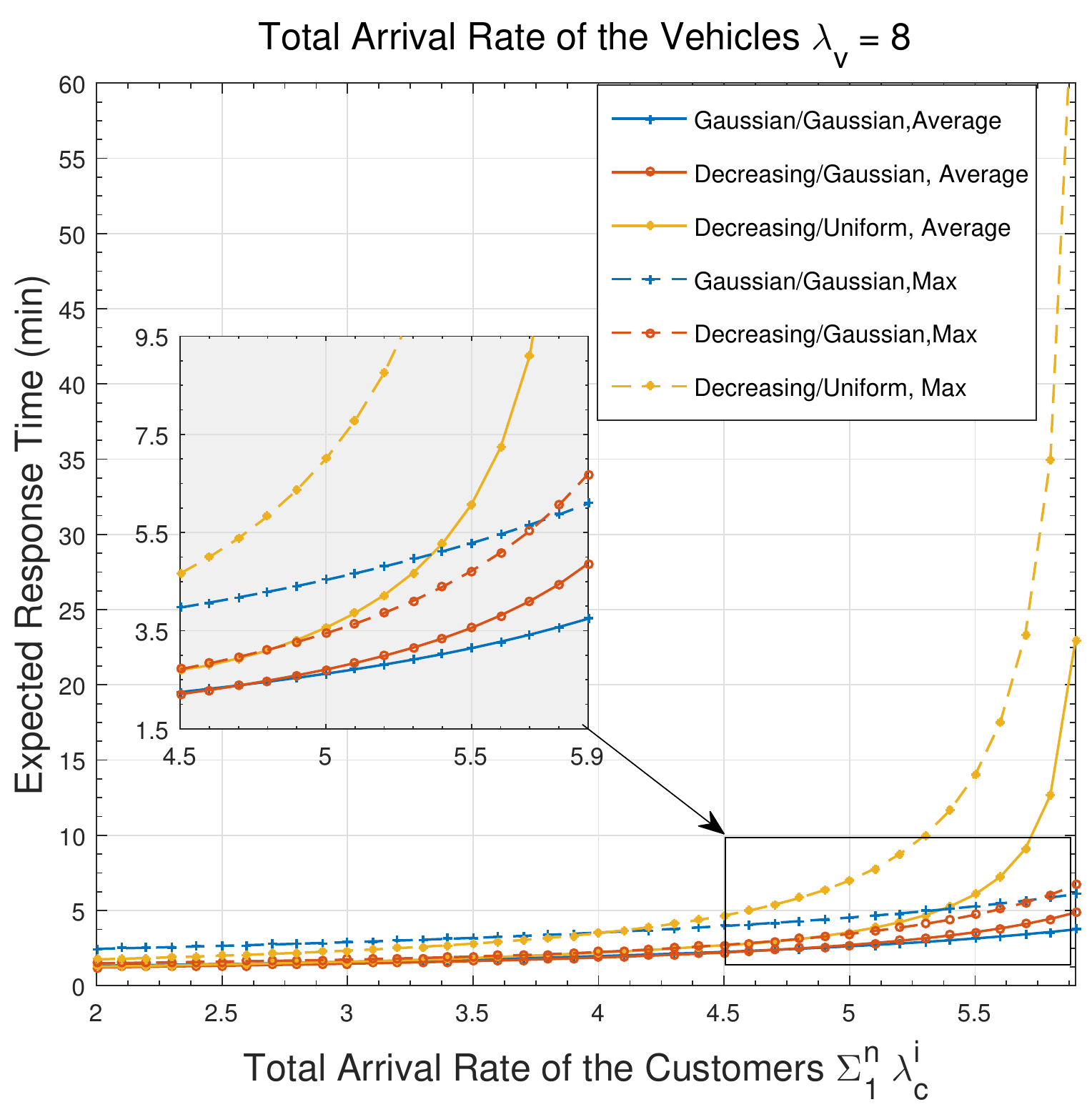}
    \caption{Effect of different customer and SoC distributions. \label{fig:comp-ditributions} }
\end{figure}
\begin{figure}[t]
\centering
  \includegraphics[width=0.77\linewidth]{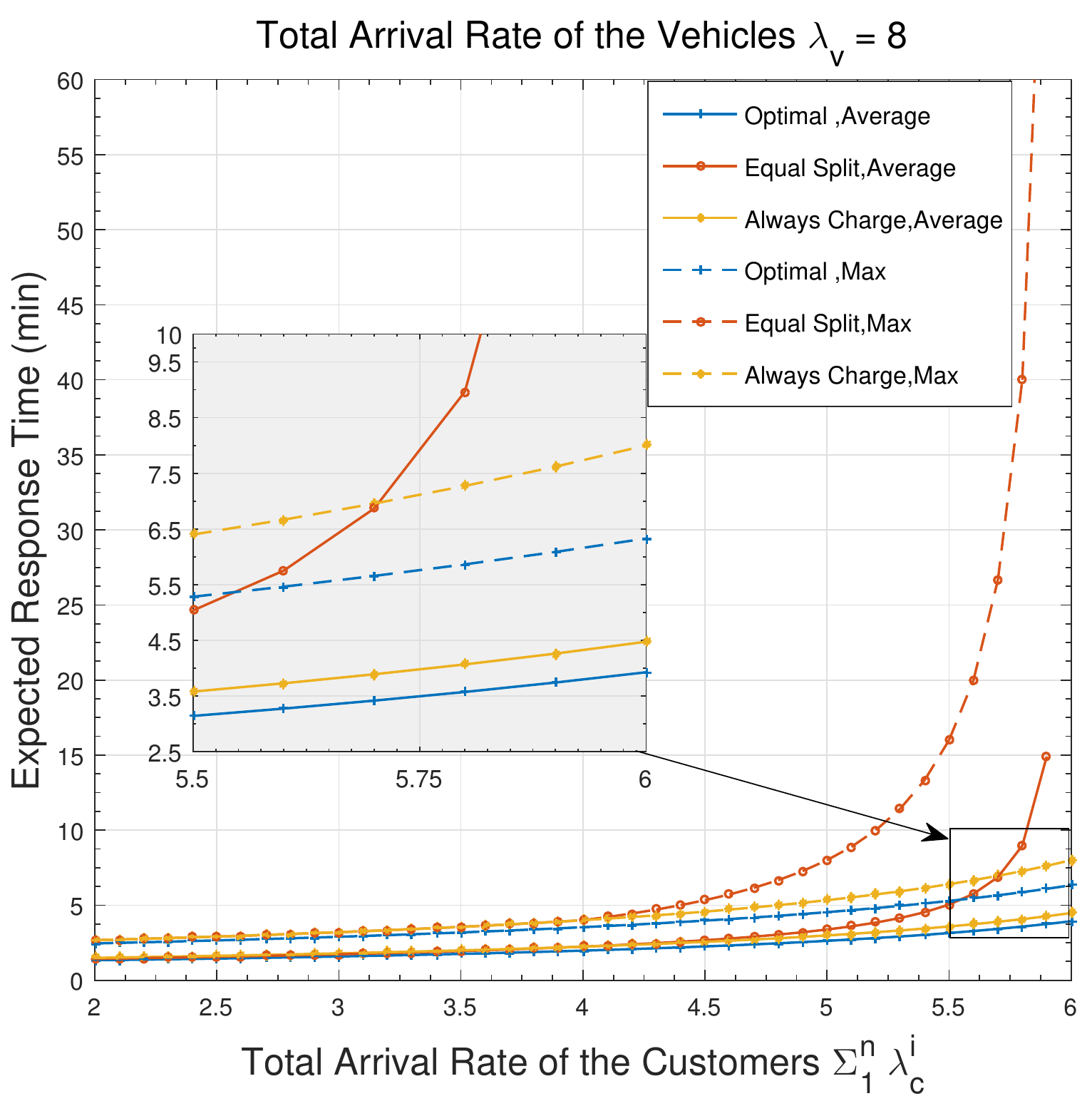}
    \caption{Comparison to non-optimized policies. \label{fig:comp-decision} }
\end{figure}
\section{Conclusion}
In this paper, we proposed solutions to the computational and charging bottlenecks threatening the success of AEMoD systems\ignore{ in attracting a large number of customers and solving private urban transportation problems}. The computational bottleneck can be resolved by employing a fog-based architecture to distribute the optimization loads over different service zones, reduce communication delays, and matches the nature of dispatching and charging processes of AEMoD vehicles. We also proposed a multi-class dispatching and charging scheme and developed its queuing model and stability conditions. We then formulated the problem of optimizing the proportions of vehicles of each class that will partially/fully charge or directly serve customers as a stochastic linear program, in order to minimize the maximum expected system response time while respecting the system stability constraints. The optimal decisions and corresponding maximum response time were analytically derived. The optimal number of classes both minimizing the response time and matching the vehicle and charging statistics was also characterized. Simulation results demonstrated both the merits of our proposed optimal decision scheme compared to typical non-optimized schemes, and its performance for different distributions of vehicle SoC and customer trip distances.
\ignore{
For the future work, we are planning to modify the system formulation and compare the final results.This will include using the optimization over the service load of the queues $\frac{\mu}{\lambda}$ or over the a whole vector $\dfrac{1}{\lambda_v^({i})-\lambda_c^({i})}$.This work will be extended to time-varying optimization and predictive control of the next time frame that will be applied to real life collected data.
The future optimization will include the number of cars depending on the demand and the area.
}

\appendices
\section{Proof of Lemma \ref{lem1}}\label{app:lem1}
From (\ref{eq:4}) and (\ref{eq:5}) we have : 
\begin{equation}\label{eq:18}
\begin{aligned}
& & & \lambda_c^{(i)} < \lambda_v(p_{i-1}\overline{q}_{i-1} + p_{i}q_{i}) , \; i = 1, \ldots, n-1.\\
& & & \lambda_c^{(n)} < \lambda_v(p_{n-1}\overline{q}_{n-1} + p_{0}q_{0}),\; i = n
\end{aligned}
\end{equation}
The summation of all the inequalities in (\ref{eq:18}) gives a new inequality 
\begin{equation}\label{eq:19}
\begin{aligned}
\sum ^{n}_{i=1}{\lambda_c^{(i)}} < \lambda_v [\sum ^{n-1}_{i=1}{(p_{i-1}\overline{q}_{i-1} + p_{i}q_{i})} + (p_{n-1}\overline{q}_{n-1} + p_{0}q_{0})]
\end{aligned}
\end{equation}
\begin{equation}\label{eq:20}
\begin{aligned}
\sum ^{n}_{i=1}{\lambda_c^{(i)}} < \lambda_v [{p_{0}\overline{q}_{0} + p_{1}q_{1}} + {p_{1}\overline{q}_{1}}+ ... + p_{n-1}\overline{q}_{n-1} + p_{0}q_{0}]
\end{aligned}
\end{equation}
We have ${\overline{q}_{i} + q_{i}}$ so ${p_{i}\overline{q}_{i} + p_{i}q_{i}} = p_{i}$
\begin{equation}\label{eq:21}
\begin{aligned}
\sum ^{n}_{i=1}{\lambda_c^{(i)}} < \lambda_v ({p_{0} + p_{1} + p_{2}}+ ... + p_{n-1})
\end{aligned}
\end{equation}
We have 
$\sum ^{n-1}_{i=0}{p_{i}} = 1  $  so  $\sum ^{n}_{i=1}{\lambda_c^{(i)}} < \lambda_v $

\section{Proof of Lemma \ref{lem2}}\label{app:lem2}
The summation of the inequalities given by (\ref{eq:6}) $\forall~i=\{0,\dots, n\}$ gives the following inequality : 
\begin{equation}\label{eq:22}
\begin{aligned}
\lambda_v\sum ^{n-1}_{i=0}p_{i} -\lambda_v\sum ^{n-1}_{i=0}p_{i}{q_{i}} +\lambda_v p_{0}q_{0} < C(n \mu_c) + \mu_c
\end{aligned}
\end{equation}
Since $\sum ^{n-1}_{i=0}p_{i}=1$ (because $p_n=0$ as described in Section 2), we get: 
\begin{equation}\label{eq:23}
\begin{aligned}
\lambda_v -\lambda_v\sum ^{n-1}_{i=1}p_{i}{q_{i}}\ignore{ +\lambda_v p_{0}q_{0}} < \mu_c (C n + 1)
\end{aligned}
\end{equation}
In the worst case, all the vehicles will be directed to partially charge before serving, which means that always  $q_i = 0$. Therefore, we get:
\begin{equation}\label{eq:24}
\begin{aligned}
C n  > \dfrac{\lambda_v}{\mu_c} - 1\;,
\end{aligned}
\end{equation}
which can be re-arranged to be: 
\begin{equation}
\begin{aligned}
n  > \dfrac{\lambda_v}{C \mu_c} - \dfrac{1}{C}
\end{aligned}
\end{equation}

\section{Proof of Theorem \ref{thm1}} \label{app:thm1}
Applying the KKT conditions to the inequalities constraints of (\ref{eq:11}), we get:
\begin{equation}\label{eq:13}
\begin{aligned}
& & &\alpha_{i}^*( {\lambda_v(p_{i-1}{q_{i-1}^*} -  p_{i}q_{i}^* ) + R^*} - \lambda_v p_{i-1} + \lambda_c^{(i)} ) = 0\\ 
& & & i= 1, \ldots, n-1.\\
& & &\alpha_{n}^* ( {\lambda_v( p_{n-1}q_{n-1}^* - p_{0}{q_{0}^*} ) + R^*} - \lambda_v p_{n-1} + \lambda_c^{(n)} ) = 0.\\
& & &\beta_{0}^* (\sum ^{n-1}_{i=0}\lambda_v(p_{i} - p_{i}{q_{i}^*}) - C (n \mu_c)) = 0.\\
& & &\beta_{1}^*(\lambda_v p_{0}q_{0}^* - \mu_c) = 0 \\
& & & \gamma_{i}^*(q_{i}^* - 1) = 0  , \;i = 0, \ldots, n-1.\\
& & &\omega_{i}^*q_{i}^* = 0 , \;i = 0, \ldots, n-1.\\
& & & \omega_{n}^*R^* = 0.
\end{aligned}
\end{equation}
Likewise, applying the KKT conditions to the Lagrangian function in (\ref{eq:12}), and knowing that the gradient of the Lagrangian function goes to $0$ at the optimal solution, we get the following set of equalities:
\begin{equation}\label{eq:14}
\begin{aligned}
& \lambda_v p_{i}(\alpha_{i+1}^* - \alpha_{i}^*) = \omega_{i}^* - \gamma_{i}^* , \;  i= 1, \ldots, n-1.\\
& \lambda_v p_{0}(\alpha_{1}^* - \alpha_{n}^*) = \omega_{0}^* - \gamma_{0}^*\\
& \sum_{i=1}^{n-1} \alpha_{i}^* = 1
\end{aligned}
\end{equation}
From Burke's theorem on the stability condition of the queues, the constraints on the charging queues are strict inequalities and the constraints on $R$ should also be strictly larger than 0. Combining the Burke's theorem and the equations on (\ref{eq:13}), we find that all the $\beta_{0}^* = \beta_{1}^* = 0$ and $\omega_{n}^* = 0$.
 
Knowing that the gradient of the Lagrangian goes to $0$ at the optimal solutions, we get the system of equalities given by (\ref{eq:14}). The fact that $\beta_{i}^* = 0$ and $\omega_{n}^* = 0$ explains the absence of  $\beta_{i}^*$ and $\omega_{n}^*$ in (\ref{eq:15}) The result given by multiplying the first equality in (\ref{eq:14}) by $q_i^*$ and the second equality by  $q_0^*$ combined with the last three equalities given by (\ref{eq:13}) gives :
\begin{equation}\label{eq:25}
\begin{aligned}
& & &\lambda_v p_{i}(\alpha_{i+1}^* - \alpha_{i}^*)q_{i}^* = - \gamma_{i}^* , \;  i= 1, \ldots, n-1.\\
& & &\lambda_v p_{0}(\alpha_{1}^* - \alpha_{n}^*)q_{0}^* = - \gamma_{0}^*\\
& & & \sum_{i=1}^{n-1} \alpha_{i}^* = 1
\end{aligned}
\end{equation}
(\ref{eq:25}) Inserted in the fifth equality in (\ref{eq:13}) gives : 
\begin{equation}\label{eq:26}
\begin{aligned}
& & &\lambda_v p_{i}(\alpha_{i+1}^* - \alpha_{i}^*)(q_{i}^* - 1)q_{i}^* = 0 , \;  i= 1, \ldots, n-1.\\
& & &\lambda_v p_{0}(\alpha_{1}^* - \alpha_{n}^*)(q_{0}^* - 1)q_{0}^* = 0\\
& & & \sum_{i=1}^{n-1} \alpha_{i}^* = 1 
\end{aligned}
\end{equation}
From (\ref{eq:26}) we have 
$0 < q_{0}^* <1$ only if $\alpha_{1}^* = \alpha_{n}^*$
And $0 < q_{i}^* <1$ only if $\alpha_{i+1}^* = \alpha_{i}^*$ 
Since $0 \leq q_{i}^* \leq 1$ then these equalities may not always be true 

if  $\alpha_{1}^* > \alpha_{n}^*$ and we know that $\gamma_{0}^* \geq 0$ then  $\gamma_{0}^* = 0$ which gives $q_{0}^* \ne 1 $ and $q_{0}^* = 0 $.

if  $\alpha_{i+1}^* > \alpha_{i}^*$ and we know that $\gamma_{i}^* \geq 0$ then  $\gamma_{i}^* = 0$ which gives $q_{i}^* \ne 1 $ and $q_{i}^* = 0 $

if  $\alpha_{1}^* < \alpha_{n}^*$ then  $\gamma_{0}^* > 0$ (it cannot be 0 because this will contradict with the value of $q_{i}$), which implies that $q_{0}^* = 1 $. 

if  $\alpha_{i+1}^* < \alpha_{i}^*$ then $\gamma_{i}^* > 0$ (it cannot be 0 because this contradicts with the value of $q_{i}$), which implies that $q_{i}^* = 1 $

Otherwise , if  $\alpha_{1}^* = \alpha_{n}^* \ne 0$ (they cannot be equal to 0 at the same time, which means that $q_{0} = 1$, and we know in advance that this cannot be the case here), we have $ q_{1}^* =\frac{p_{0}{q_{0}^*}}{p_{1}} - \frac{\lambda_v p_{0} - \lambda_c^{(1)} - R^*}{\lambda_v p_{1} } \; $ and $ q_{n-1}^* =\frac{p_{0}{q_{0}^*}}{p_{n-1}} - \frac{\lambda_v p_{0} - \lambda_c^{(n)} - R^*}{\lambda_v p_{n-1} } \; $

Finally, if  $\alpha_{i+1}^* = \alpha_{i}^* \ne 0$ (they cannot be equal to 0 at the same time, which means that $q_{i} = 1$, and we know in advance that this cannot be the case here), we have $ q_{i}^* =\frac{p_{i-1}{q_{i-1}^*}}{p_{i}} - \frac{\lambda_v p_{i-1} - \lambda_c^{(i)} - R^*}{\lambda_v p_{i} } $ and $ q_{i+1}^* =\frac{p_{i}{q_{i}^*}}{p_{i+1}} - \frac{\lambda_v p_{i} - \lambda_c^{(i+1)} - R^*}{\lambda_v p_{i+1} }  $

\section{Proof of Theorem \ref{thm2}} \label{app:thm2}
To prove this theorem, we first start by putting the problem on the standard linear programming form as follows:
\begin{equation}\label{eq:27}
\begin{aligned}
& \underset{q_0,q_1,\ldots ,q_{n-1}}{\text{minimize }} -R \\
& \text{subject to} \\
&\text{Constraint on costumers arrivals queues}\\
& {\lambda_v(p_{i-1}{q_{i-1}} -  p_{i}q_{i} ) + R} \leq \lambda_v p_{i-1} - \lambda_c^{(i)} , \; i = 1, \ldots, n-1.\\
& {\lambda_v( p_{n-1}q_{n-1} - p_{0}{q_{0}} ) + R} \leq \lambda_v p_{n-1} - \lambda_c^{(n)} .\\
& \text{Constraint on charging vehicles queues}\\
&-\lambda_v\sum ^{n-1}_{i=0} p_{i}{q_{i}^*} < C (n \mu_c) - \lambda_v\\
&\lambda_v p_{0}q_{0} < \mu_c \\
&\text{Constraint on probabilities and decisions}\\
&q_{i} \leq 1  , \; i = 0, \ldots, n-1.\\
&-q_{i} \leq 0 , \; i = 0, \ldots, n-1.\\
& -R < 0.\\
&\sum ^{n-1}_{i=0}p_i = 1 , \; 0 < p_{i} < 1  , \; i = 0, \ldots, n-1. 
\end{aligned}
\end{equation}
Writing the problem on its matrix form, we get: 
\begin{equation}\label{eq:39}
\begin{aligned}
& \underset{X}{\text{minimize}}
& &  \mathbf{c}^T \mathbf{x} \\
& \text{subject to} & & \mathbf{Ax} \preceq \mathbf{b}\\
\end{aligned}
\end{equation}
where:
\begin{equation}
\begin{aligned}
&\mathbf{x}_{(n+1\times 1)}=\begin{pmatrix}
    q_0 \\ q_1 \\ \vdots \\ q_{n-1} \\ R
\end{pmatrix}
&\mathbf{c}_{(n+1\times 1)}=\begin{pmatrix}
    0 \\ 0 \\ 0 \\ \vdots \\ -1
\end{pmatrix}\\
&\mathbf{b}_{(3n+4\times 1)}=\begin{pmatrix}
   \lambda_v p_0 - \lambda_c^{(1)} \\ \vdots \\ \lambda_v p_{n-1} - \lambda_c^{(n)} \\C (n \mu_c) - \lambda_v \\ \mu_c\\ 1 \\ \vdots \\  1 \\ \infty \\ 0 \\ \vdots \\  0 
\end{pmatrix}\\
\end{aligned}
\end{equation}

\begin{equation}
\begin{aligned}
&\mathbf{A}_{(3n+4\times n+1)}= \\
&\begin{small}
\begin{pmatrix}
   \lambda_v p_0 & -\lambda_v p_1 & 0 & \dots & 0 & 1 \\
	0 & \lambda_v p_1 & -\lambda_v p_2  & \dots & 0 & 1 \\ 
\vdots & \ddots & \ddots & \ddots & \ddots & \vdots \\
    0 & \dots & 0 & \lambda_v p_{n-2} & -\lambda_v p_{n-1} & 1 \\
	-\lambda_v p_0 &  0 & \dots & \dots & \lambda_v p_{n-1} & 1 \\
  - \lambda_v p_0 &  - \lambda_v p_1 & \dots & \dots &  - \lambda_v p_{n-1} & 0  \\
    \lambda_v p_0 &0 & \dots & \dots & \dots & 0  \\
& &{\huge\mbox{{$I_{n+1}$}}} &   &   \\
& & &  &  &  &   \\
& &{\huge\mbox{{$-I_{n+1}$}}} &   &   \\
\end{pmatrix}
\end{small}
\end{aligned}
\end{equation}

The matrix form of the Lagrangian function can be thus expressed as:
 Lagrangian :\\
\begin{equation}
\begin{aligned}
L(\mathbf{x},\boldsymbol{\nu}) = \mathbf{c}^T \mathbf{x} + \boldsymbol{\nu}^T (\mathbf{Ax}-\mathbf{b}) = -\mathbf{b}^T+(\mathbf{A}^T\boldsymbol{\nu} + \mathbf{c})^T \mathbf{x}\;,\\
\end{aligned}
\end{equation}
where $\boldsymbol{\nu}$ is the vector of the dual variables or Lagrange multipliers vector associated with the problem \ref{eq:39}. Each element $\nu_{i}$ of $\boldsymbol{\nu}$ is the Lagrange multiplier associated with the $i$-th inequality constraint $\mathbf{a}_{i}\mathbf{x} - b_{i} \leq 0$, where $\mathbf{a}_i$ and $b_i$ are the the $i$-th row and and $i$-th element of matrix $\mathbf{A}$ and vector $\mathbf{b}$, respectively. In fact $\boldsymbol{\nu}$ is the vector that includes all the vectors $\boldsymbol{\alpha}$, $\boldsymbol{\beta}$, $\boldsymbol{\gamma}$, $\boldsymbol{\omega}$ as follows: 
\begin{equation}
\boldsymbol{\nu}^T_{(1\times 3n+4)}= \begin{pmatrix} \alpha_1  \; \dots \; \alpha_n \; \beta_0 \; \beta_1 \;  \gamma_0 \; \dots \; \gamma_n \; \omega_0 \; \dots \; \omega_{n}\end{pmatrix}
\end{equation}
We will used this combined notation for ease and clarity of notation. 

The Lagrange dual function is expressed as:
\begin{equation}
\begin{aligned}
& g(\boldsymbol{\nu}) = \underset{\mathbf{x}}\inf~ L(\mathbf{x},\boldsymbol{\nu}) = -\mathbf{b}^T \boldsymbol{\nu} + \underset{\mathbf{x}}\inf~(\mathbf{A}^T\boldsymbol{\nu} +\mathbf{c})^T\mathbf{x}\;,
\end{aligned}
\end{equation}
The solution for this function is easily determined analytically, since a linear function is bounded below only when it is identically zero. Thus, $ g(\boldsymbol{\nu}) = -\infty$  except when $\mathbf{A}^T\boldsymbol{\nu} + \mathbf{c} = \mathbf{0}$, where $\mathbf{0}$ is the all zero vector. Consequently, we have:
\begin{equation}
\begin{aligned}
& g(\boldsymbol{\nu}) =  \begin{cases}
            -\mathbf{b}^T \boldsymbol{\nu} & \mathbf{A}^T\boldsymbol{\nu} + \mathbf{c} = \mathbf{0} \\
            -\infty & \text{ortherwise}\\
            \end{cases} \\
\end{aligned}
\end{equation}
For each $\nu\succeq \mathbf{0}$ (i.e., $\nu_{i} \geq 0~\forall~i$), the Lagrange dual function gives us a lower bound on the optimal value of the original optimization problem. This leads to a new equivalent optimization problem, which is the dual problem:
\begin{equation}
\begin{aligned}
& \underset{\boldsymbol{\nu}}{\text{maximize}}
& &  g(\boldsymbol{\nu}) = -\mathbf{b}^T\boldsymbol{\nu} \\
& \text{subject to} & & \mathbf{A}^T\boldsymbol{\nu} + \mathbf{c} = \mathbf{0} \\
& & & \boldsymbol{\nu} \succeq \mathbf{0} 
\end{aligned}
\end{equation}
Applying Slater's Theorem for duality qualification, and since strong duality holds for the considered optimization problem, then solving the dual problem gives the exact optimal solution for the primal problem. This is described by the equality :
\begin{equation}
\begin{aligned}
& & g(\boldsymbol{\nu}^*) = -\mathbf{b}^T \boldsymbol{\nu}^* = \mathbf{c}^T \mathbf{x}^* = -R^*
\end{aligned}
\end{equation}
By expanding on the values of $\mathbf{b}$ and $\boldsymbol{\nu}$ in the above equation, the optimal value of $R^*$ can be expressed as:
 \begin{equation}
R^* = \sum_{i=1}^{n} (\lambda_v p_{i-1} - \lambda_c^{(i)}) \alpha_{i}^* + \sum_{i=0}^{n-1} \gamma_i^*
\end{equation}

\end{document}